%% file: main.tex
\documentclass{llncs}
 
\usepackage{microtype}

\usepackage{amssymb}
\usepackage{amsmath}
\usepackage{wrapfig}
\usepackage{epsfig}
\usepackage[all]{xy}
\usepackage{proofs}
\usepackage{amsmath}
\usepackage{amssymb}
\usepackage{stmaryrd}
\usepackage{float}
\usepackage{fancyvrb}
\usepackage{etex}
\usepackage{tikz}
\usepackage{relsize}
\usepackage{url}

\usepackage{mathtools}
\DeclarePairedDelimiter{\ceil}{\lceil}{\rceil}

\usepackage{listings}
\floatstyle{ruled}
\restylefloat{table}
\restylefloat{figure}

\def\|#1|{\mathid{#1}}
\newcommand{\mathid}[1]{\ensuremath{\mathit{#1}}}
\def\codesize{\smaller}
\def\<#1>{\codeid{#1}}
\newcommand{\codeid}[1]{\ifmmode{\mbox{\codesize\ttfamily{#1}}}\else{\codesize\ttfamily #1}\fi}

\newcommand{\ttlcb}{\texttt{\char "7B}}
\newcommand{\ttrcb}{\texttt{\char "7D}}

\usepackage{basic}

\begin{document}

\title{Semantics for Locking Specifications}
\titlerunning{Semantics for Locking Specifications} 

\author{Michael D. Ernst\inst{1} \and Damiano Macedonio\inst{2} \and Massimo Merro\inst{3} \and Fausto Spoto\inst{2,3}}
\institute{Computer Science \& Engineering, University of Washington, WA, USA \and Julia Srl, Verona, Italy \and Dipartimento di Informatica, Universit\`a degli Studi di Verona, Italy}

\authorrunning{Michael  Ernst et al.} 

\maketitle

\begin{abstract}
To prevent concurrency errors, programmers need to obey a locking 
discipline. Annotations that specify that discipline, such as Java's
\<@GuardedBy>, are already widely used.  Unfortunately, their
semantics is expressed informally and is consequently ambiguous.
This article highlights such ambiguities and
overcomes them by formalizing two possible semantics of \<@GuardedBy>,
using a reference operational semantics for a
core calculus of a concurrent Java-like language.
It also identifies when such annotations are actual guarantees against
data races. Our work aids in understanding the annotations and supports the
development of sound tools that verify or infer them.
\looseness=-1
\end{abstract}

\input{introduction}

\input{examples}

\input{syntax}

\input{semantics}

\input{guardedBy}
\input{protection}

\input{julia}

\input{conclusion}

\bibliography{biblio,bib/bibstring-abbrev,bib/ernst}

\appendix

\section{Proofs of Sec~\ref{subsec:protection}}

\begin{theorem}[Name-protection semantics 
vs.\ data race protection]
Let $E$ be an expression in a program, and 
$x$ be a non-aliased variable or field that is name protected by {\em $\guardedby{E}$}.
If $x$ is a variable, let $E$ contain no variable distinct from {\em \<itself>\/};
if $x$ is a field, let $E$ contain no variable distinct from {\em \<itself>} and {\em \<this>\/}.
Then, no data race can occur at those locations bound to $x$, at
any execution trace of that program.
\end{theorem}
\begin{proof} The proof is by contradiction. 
Let $
  \anglepair{T_0}{\mu_0} \reduces{}^{\ast} \anglepair{T}{\mu}
$
be an arbitrary trace of our program, where 
$T_i=\thread{\record{k_i.m_i}{C_i}{\sigma_i}\stsep S_i}{\lockset_i}$ is the $i$-th thread of $T$. 
By Def.~\ref{def:data_race}, 
if a data race occurred in $\anglepair{T}{\mu}$, at some location $l$,
bound to $x$, then $ \anglepair{T}{\mu}$  could evolve 
in at least two ways, say
\begin{center}
$\anglepair{T}{\mu} \trans{a} \anglepair{T'}{\mu'}$ and 
$\anglepair{T}{\mu} \trans{b} \anglepair{T''}{\mu''}$
\end{center}
that dereference $l$ in two different threads $a$ and $b$. 
As $x$ is non-aliased, by Def.~\ref{def:aliasing} it cannot be used in one 
thread as a variable and in the other as a field. 
As $a\neq b$, by Def.~\ref{def:aliasing} the name $x$ cannot be 
used in both threads as local variable bound to the same location $l$. Thus, there
is only one possibility:  
$x$ is a field accessed by both threads, $a$ and $b$, to dereference the location $l$. 
This means that there exist two expressions $E_a$ and $E_b$ 
such that ${E_a}.x \in \acc{C_a}$ and ${E_b}.x \in \acc{C_b}$.
As $x$ is non-aliased, by Def.~\ref{def:aliasing} 
there cannot be two different containers (objects) 
of the same field $x$. 
As a consequence, the two expressions $E_a$ and 
$E_b$ must evaluate to the same value, \emph{i.e.\/} 
$\eval{E_a}{\mu}{\sigma_a}=\eval{E_b}{\mu}{\sigma_b}=\anglepair{l'}{{\mu'}}$, for some $l'$ and $\mu'$. We recall that $\lockset_a$ and $\lockset_b$
denote the set of locations locked in $\anglepair{T}{\mu}$ by thread $a$ and $b$, respectively.  
As $x$ is name protected by \<@GuardedBy(>$E$\<)>, Def.~\ref{def:guardedByFields} 
entails that $
\loc{\eval{E}{\mu'}{\sigma_a[{\mathtt{this}}\mapsto {l'}][\mathtt{itself}\mapsto {l''}]}} \in \lockset_a
$  and 
$\loc{\eval{E}{\mu'}{\sigma_b[{\mathtt{this}}\mapsto {l'}][\mathtt{itself}\mapsto {l''}]}} \in \lockset_b
$, for $l''= \state{\mu'(l')}(x)$. 
As $x$ is a field, by hypothesis the guard expression $E$ may only 
contain the variables \<this> and \<itself>. As a consequence, 
$
\loc{\eval{E}{\mu'}{\sigma_a[{\mathtt{this}}\mapsto {l'}][\mathtt{itself}\mapsto {l''}]}}$
$ =$
$ \loc{\eval{E}{\mu'}{[{\mathtt{this}}\mapsto {l'}][\mathtt{itself}\mapsto {l''}]}}$
$=$
$
\loc{\eval{E}{\mu'}{\sigma_b[{\mathtt{this}}\mapsto {l'}][\mathtt{itself}\mapsto {l''}]}} 
$. 
By Prop.~\ref{prop:lock-threads} this is not possible as two  threads cannot lock the same location
at the same time.
\end{proof}

The requirement on the absence of aliasing is not necessary when 
working with a value-protection semantics. 
\begin{theorem}[Value-protection semantics of \<@GuardedBy> vs.\ data race protection]
Let $E$ be an expression in a program, and
$x$ be a variable or field that is value-protected by {\em \<@GuardedBy($E$)>\/}.
Let $E$ contain no variable distinct from {\em \<itself>\/}.
Then no data race can occur at those locations bound to $x$, during 
any execution trace of the program. 
\end{theorem}
\begin{proof}
Again, the proof is by contradiction. Let $
  \anglepair{T_0}{\mu_0}
    \reduces{n_0} \cdots  \reduces{n_{i-1}} 
\anglepair{T_i}{\mu_i}
$ be an arbitrary trace of our program. 
If a data race occurred at a location $l$ bound to $x$ 
in that trace, then that trace could evolve 
in at least two ways, both dereferencing $l$ but in distinct threads, say 
$a$ and $b$ (Def.~\ref{def:data_race}). 
Since $x$ is value protected by \<@GuardedBy($E$)>, both threads
lock the value (\textit{i.e.}, location) of $E$. Formally, 
by Def.~\ref{def:gardedByValueVariables} and~\ref{def:gardedByValueFields}
it follows that
$\loc{\eval{E}{\mu_i}{\sigma_i^{a}[\<itself>\mapsto l]}} \in \lockset_i^{a}$
and 
$\loc{\eval{E}{\mu_i}{\sigma_i^{b}[\<itself>\mapsto l]}} \in \lockset_i^{b}$. 
Since \<itself> is the only variable allowed in $E$,
environments $\sigma_i^a$ and $\sigma_i^b$ are irrelevant in these evaluations
and
$\loc{\eval{E}{\mu_i}{\sigma_i^{a}[\<itself>\mapsto l]}}=
\loc{\eval{E}{\mu_i}{[\<itself>\mapsto l]}}=
\loc{\eval{E}{\mu_i}{\sigma_i^{b}[\<itself>\mapsto l]}}$.
Again, by Prop.~\ref{prop:lock-threads}, this is impossible, as  
two  threads cannot lock the same location at the same time.   
\end{proof}

\input{semantics_properties}

\end{document}

%% file: introduction.tex
\section{Introduction}\label{sec:introduction}

Concurrency can increase program performance by
scheduling parallel independent tasks on multicore hardware, and
can enable responsive user interfaces.
However, concurrency might induce problems such as \emph{data races}, \textit{i.e.},
concurrent access to shared data by different threads,
with consequent unpredictable or erroneous software behavior.
Such errors are difficult to understand, diagnose, and reproduce.
They are also difficult to prevent: testing tends to be incomplete
due to nondeterministic scheduling choices made by the
runtime, and model-checking scales poorly to real-world code.

The simplest approach to prevent data races is to follow a
\emph{locking discipline} while accessing shared data: always hold a given lock when accessing  
a datum.  It is easy to violate the locking discipline, so
tools that verify adherence to the
discipline are desirable. These tools require a \emph{specification
language} to express the intended locking discipline. 
The focus of this paper is on the formal definition of a specification 
language, its semantics, and the guarantees that it gives against data races.

In Java, the most popular  
specification language 
for expressing a locking discipline is the
\<@GuardedBy> annotation.  
Informally, if the programmer annotates a field or variable $f$ as $\guardedby{E}$ 
then a thread may access $f$ only while holding the monitor corresponding to 
the \emph{guard expression} 
$E$. 
The \<@GuardedBy> annotation was
proposed by Goetz~\cite{GoetzPBB06} as a documentation
convention only, without tool support.
The annotation has been
adopted by practitioners; 
GitHub contains about
35,000 uses of the annotation in 7,000 files.
Tool support now exists in
Java PathFinder~\cite{JPF}, the Checker Framework~\cite{DietlDEMS11},
IntelliJ~\cite{Pech10}, and Julia~\cite{Julia}.  
\looseness=-1

All of these tools  rely on the previous informal definition of $\guardedby{E}$~\cite{JavadocGuardedBy}. 
However, such an informal description is prone to 
many ambiguities. 
Suppose a field $f$ is annotated as $\guardedby{E}$, for some guard 
expression 
$E$.
\label{four-ambiguities}
(1) The definition above does not clarify how an occurrence of 
the special variable \<this>
\footnote{Normally, \<this> denotes the Java reference to the current object.}
 in
$E$ should be interpreted in client code; 
this actually depends on the context
in which  $f$ is accessed.
(2) It does not define what an \emph{access} is.
(3) It does not say whether the lock statement must use
the guard expression $E$ as written 
or whether
a different expression that evaluates to the same value is permitted.
(4)  It does not indicate whether the lock that must be taken is $E$
at the time of synchronization or $E$ at the
time of field access: side effects on $E$ might make a difference here.
(5) It does not clarify whether the lock on the guard $E$ must be taken
when accessing the field \emph{named} $f$ or the \emph{value} bound to  $f$.

The latter ambiguity is particularly important. 
The interpretation of \<@GuardedBy> based on names is 
adopted in most tools
appearing in the literature~\cite{JPF,Pech10,Julia}, whereas the interpretation 
based on values seems to be less common~\cite{DietlDEMS11,Julia}. 
As a consequence, 
it is interesting to understand whether and how these two possible
interpretations of \<@GuardedBy> 
actually protect against data races on
the annotated field/variable.

The main contribution of this article is the formalization of 
two different semantics for annotations of the form  \<@GuardedBy(\|E|) \emph{Type} x>: a
\emph{name-protection} semantics, in which accesses to the annotated
\emph{name} \<x> need to be synchronized on the guard expression $E$, and a 
\emph{value-protection} semantics, in which accesses to 
a \emph{value} referenced by \<x> need to be synchronized on $E$.
The semantics clarify all the above ambiguities, so that programmers and tools
know what those annotations mean and which guarantees they entail. 
We then show that both the name-protection and the value-protection 
semantics can protect against data races under proper restrictions on 
the variables occurring in the guard expression $E$.
The name-protection semantics requires a further constraint ---
the protected variable or field must not be aliased.
Our formalization relies on a reference semantics for a concurrent fragment of
Java, which we provide in the
\emph{structural operational semantics} style of Plotkin~\cite{Plo81}.

\begin{figure}[t]
{\scriptsize\begin{lstlisting}[firstnumber=47,xleftmargin=8ex]
public final class ExecutionList {
\end{lstlisting}\vspace{-1.5em}
\begin{lstlisting}[firstnumber=66,xleftmargin=8ex]
  private @GuardedBy(this) RunnableExecutorPair runnables;
\end{lstlisting}\vspace{-1.5em}
\begin{lstlisting}[firstnumber=116,xleftmargin=8ex]
  public void execute() {
    RunnableExecutorPair list;
    synchronized (this) {
\end{lstlisting}\vspace{-1.5em}
\begin{lstlisting}[firstnumber=123,xleftmargin=8ex]
      list = runnables;
      runnables = null;
    }
\end{lstlisting}\vspace{-1.5em}
\begin{lstlisting}[firstnumber=137,xleftmargin=8ex]
    RunnableExecutorPair reversedList = null;
    while (list != null) {
      RunnableExecutorPair tmp = list;
      list = list.next;
      tmp.next = reversedList;
      reversedList = tmp;
    }
\end{lstlisting}\vspace{-1.5em}
\begin{lstlisting}[firstnumber=148,xleftmargin=8ex]
  }
\end{lstlisting}\vspace{-1.5em}
\begin{lstlisting}[firstnumber=177,xleftmargin=8ex]
}
\end{lstlisting}}
\vspace*{-2mm}
\caption{Guava 18's \<com.google.common.util.concurrent.ExecutionList> class.
The \<@GuardedBy> annotation (line 66) is satisfied for name, but
not for value, protection.
\looseness=-1
}
\label{fig:example_from_Guava}
\end{figure}
We have used our formalization to extend  the Julia static analyzer~\cite{Julia} to
check and infer \<@GuardedBy> annotations in arbitrary Java code.
Julia allows the user to select either name-protection or value-protection. 
Our implementation reveals that most programmer-written \<@GuardedBy> annotations do not satisfy either
of the two alternative semantics given in this paper. 
For instance, in the code
of Google Guava~\cite{Guava} (release $18$),
the programmer put $64$ annotations on fields; $17$ satisfy the semantics of
name protection;  $9$ satisfy the semantics of value protection; the others
do no satisfy any of the two.
Fig.~\ref{fig:example_from_Guava} shows an example of an annotation written
by the programmers of Guava and that statisfy only the name protection.
Namely, field \<runnables> is annotated as \<@GuardedBy(this)>
but its value is accessed without synchronization at line \<140>~\cite{Guava}.
In this extended abstract proofs are omitted; they can be found in the appendix.

\paragraph*{Outline.} Sec.~\ref{sec:examples}
discusses the informal semantics of \<@GuardedBy> by way of examples.
Sec.~\ref{sec:syntax} defines the syntax and semantics of
a concurrent fragment of Java.
Sec.~\ref{sec:guardedBy} gives formal definitions for both the
name-protection and value-protection semantics. Sec.~\ref{subsec:protection} shows 
which guarantees they provide against data races.
Sec.~\ref{sec:julia} describes the implementation in Julia.
Sec.~\ref{sec:conclusion} discusses related work and concludes.


%% file: examples.tex
\section{Informal Semantics of \<@GuardedBy>}
\label{sec:examples}

This section illustrates the use of \<@GuardedBy> by example. 
Fig.~\ref{fig:running_example} defines an observable object that allows clients to
concurrently register listeners. Registration must be synchronized
to avoid data races:  simultaneous modifications of the \<ArrayList>
might result in a corrupted list or lost registrations.
Synchronization is needed in the \<getListeners()> method as well,
or otherwise the Java memory model does not guarantee the inter-thread visibility 
of the registrations.

The interpretation of the \<@GuardedBy(this)> annotation on field
\<listeners> requires resolving the ambiguities explained in Sec.~\ref{sec:introduction}.
The intended locking discipline is that every use of \<listeners> should be enclosed within
a construct \<synchronized (\|container|) \ttlcb...\ttrcb>, where \|container|
denotes the object whose field \<listeners> is accessed (ambiguities (1) and (2)).
For instance, the access \<original.listeners> in the copy constructor
is enclosed within \<synchronized (original) \ttlcb...\ttrcb>.
This contextualization \emph{of the guard expression\/}, similar to viewpoint
adaptation~\cite{DietlDM2007},
is not clarified in any informal definitions of \<@GuardedBy> (ambiguity (3)).
Furthermore, 
it is not clear if a definite alias of \<original> can be used as synchronization guard
at line $5$. It is not clear if \<original> would be allowed to be reassigned between lines
$5$ and $6$ (ambiguity (4)).
Note that the copy constructor does not synchronize on \<this> even though
it accesses \<this.listeners>.  This is safe so long as the constructor
does not leak \<this>.  This paper assumes that an escape
analysis~\cite{Blanchet03} has established that constructors do not leak \<this>.
The \<@GuardedBy(this)> annotation on field \<listeners> suffers also 
from ambiguity (5): it is not obvious whether it intends to protect
the \emph{name} \<listeners> (\emph{i.e.}, the name can be only used when the
lock is held) or the value currently bound to \<listeners> (\emph{i.e.}, that
value can be only accessed when the lock is held).
Another way of stating this is that \<@GuardedBy> can be interpreted as
a \emph{declaration annotation} (a restriction on uses of a name) or as
a \emph{type annotation} (a restriction on values associated to that name). 
\looseness=-1

\begin{figure}[t]
\vspace*{-1mm}
{\scriptsize\begin{lstlisting}[xleftmargin=8ex]
public class Observable {
  private @GuardedBy(this) List<Listener> listeners = new ArrayList<>();
  public Observable() {}
  public Observable(Observable original) {  // copy constructor
    synchronized (original) {
      listeners.addAll(original.listeners);
    }
  }
  public void register(Listener listener) {
    synchronized (this) {
      listeners.add(listener);
    }
  }
  public List<Listener> getListeners() {
    synchronized (this) {
      return listeners;
    }
  }
}
\end{lstlisting}}
\vspace*{-2mm}
\caption{This code has a potential data race due to aliasing of the
  \<listeners> field.}
\label{fig:running_example}
\end{figure}

The code in Fig.~\ref{fig:running_example} seems to satisfy
the name-protection locking discipline expressed by the annotation \<@GuardedBy(this)> for
field \<listeners>: every use of
\<listeners> occurs in a program point where the current thread locks its
container, and we conclude that \<@GuardedBy(this)> name-protects \<listeners>.
Nevertheless, a data race is possible, since 
two threads could call \<getListeners()> and later access the returned
value concurrently. 
This is inevitable when critical sections \emph{leak} guarded data. More generally,
name protection does not prevent data races
if there are aliases of the guarded name (such as a returned value in our example) that
can be used in an unprotected manner.
The value-protection semantics of \<@GuardedBy> is not affected by aliasing
as it tracks accesses to
the value referenced by
the name, not the name itself.

Any formal definition of \<@GuardedBy>
must result in mutual exclusion in order to ban data races.
If $x$ is \<@GuardedBy($E$)>, then at every
program point \textit{P} where a thread accesses $x$
(or its value),  the thread must hold the lock on $E$.
Mutual exclusion
%
%
requires that two conditions are satisfied:
(i) $E$ can be evaluated at all program points \textit{P},
and
(ii) these evaluations always yield the same value.

Point (i) is syntactic and related to the fact that $E$ cannot refer to
variables or fields that are not always in scope or visible at all program points \textit{P}.
This problem exists for both name protection and value protection, but is more
significant for the latter, that is meant to protect values that flow in the
program through arbitrary aliasing.
For instance, the annotation \<@GuardedBy(listeners)> cannot be used
for value protection in Fig.~\ref{fig:running_example},
since the name \<listeners> is not visible outside class \<Observable>, but its
value flows outside that class through method \<getListeners()> and must be protected
also if it accessed there.
The value protection semantics supports
a special variable \<itself> in $E$,
that refers to the current value of $x$ being protected,
without problems of scope or visibility.
For instance, for value protection, the code in Fig.~\ref{fig:running_example}
could be rewritten as in Fig.~\ref{fig:running_example_by_value}.

\begin{figure}[t]
\vspace*{-1mm}
{\scriptsize\begin{lstlisting}[xleftmargin=8ex]
public class Observable {
  private @GuardedBy(itself) List<Listener> listeners = new ArrayList<>();
  public Observable() {}
  public Observable(Observable original) {     // copy constructor
    synchronized (original.listeners) {
      listeners.addAll(original.listeners);
    }
  }
  public void register(Listener listener) {
    synchronized (listeners) {
      listeners.add(listener);
    }
  }
  public  List<Listener> getListeners() {
    synchronized (listeners) {
      return listeners;
    }
  }
}
\end{lstlisting}}
\vspace*{-2mm}
\caption{Value protection prevents data races;
see \<itself> in the guard expression.}
\label{fig:running_example_by_value}
\end{figure}

Point (ii) is semantical and related to the intent of
providing a guarantee of mutual exclusion.
For instance,
in Fig.~\ref{fig:running_example_by_value}, value protection bans data races on
\<listeners> since the guard \<itself> can be evaluated everywhere (point (i)) and always yields
the value of 
\<listeners> itself (point (ii)).
Here, the \<@GuardedBy(itself)> annotation
requires all accesses to the value of \<listeners> to occur
only when the current thread locks the same monitor --- even outside class
\<Observable>, in a client that operates on the value returned by \<getListeners()>.
In Fig.~\ref{fig:changing_lock}, instead, field
\<listeners> is \<@GuardedBy(guard)> according to both
name protection and value protection, but the value of \<guard>
is distinct at different program points: no mutual exclusion
guarantee exists and data races on \<listeners> occur.
\begin{figure}[t]
\vspace*{-1mm}{\scriptsize\begin{lstlisting}[xleftmargin=8ex]
public class Observable {
  private @GuardedBy(guard) List<Listener> listeners = new ArrayList<>();
  private Object guard1 = new Object();
  private Object guard2 = new Object();
  public Observable() {}
  public Observable(Observable original) {     // copy constructor
    Object guard = guard1;
    synchronized (guard) {
      listeners.addAll(original.listeners);
    }
  }
  public void register(Listener listener) {
    Object guard = guard2;
    synchronized (guard) {
      listeners.add(listener);
    }
  }
}
\end{lstlisting}}
\vspace*{-2mm}
\caption{If the guard expression refers to distinct values at distinct
  program points, concurrent accesses to \<listeners> can race.
}
  \label{fig:changing_lock}
\end{figure}

%% file: syntax.tex
\section{A Core Calculus for Concurrent Java}\label{sec:syntax}
Some preliminary notions are needed to define our calculus. 
A \emph{partial function} $f$ from $A$ to $B$ is denoted by $f : A \rightharpoonup B$, and
its \emph{domain} is $\dom{f}$.
We write $f(v)\!\downarrow$ if $v\in\dom{f}$ and $f(v)\!\uparrow$ otherwise.
The symbol $\phi$ denotes the empty function, such that $\dom{\phi} = \emptyset$;
$\mapdef{v_1\mapsto t_1, \ldots, v_n\mapsto t_n}$ denotes the function $f$ with
$\dom{f} = \{v_1,\ldots,v_n\}$ and $f(v_i) = t_i$ for $i = 1,\ldots, n$;
$f[v_1\mapsto t_1, \ldots, v_n\mapsto t_n]$ denotes the update of $f$, 
where $\dom{f}$ is enlarged for every $i$ such that $v_i\notin\dom{f}$.
A tuple is denoted as $\langle v_0,\ldots, v_n\rangle$. 
A \emph{poset} is a structure $\langle A,\leq \rangle$ where $A$ is a set with a reflexive,
transitive, and antisymmetric relation $\leq$. 
Given $a\in A$, we define
$\uparrow a \eqdef \{a' : a\leq a'\}$. A \emph{chain} is a totally ordered poset.
\subsection{Syntax}\label{subsec:syntax}
Symbols $f,g,x,y, \ldots$ range over a set of variables
$\mathit{Var}$ that includes $\this$.
Variables identify either local variables
in methods or instance variables (\emph{fields}) of objects. 
Symbols $m,p, \ldots$ range over a set $\mathit{MethodName}$ of method names. 
There is a set $\locations$ of memory locations, ranged over by $l$.
Symbols $\kappa, \kappa_0, \kappa_1,\ldots $ range over a set of \emph{classes} (or 
\emph{types}) $\classes$, ordered by a \emph{subclass relation} $\leq$;
$\anglepair{\classes}{\leq}$ is a poset such that for all 
$\kappa\in\classes$ the set ${\uparrow\!\kappa}$ is a finite chain.
Intuitively, $\kappa_1\leq \kappa_2$ means that $\kappa_1$ is a
\emph{subclass} (or \emph{subtype}) of $\kappa_2$.
If $m\in\mathit{MethodName}$, then $\kappa.m$ denotes the implementation of
$m$ inside class $\kappa$, if any.
The partial function $\lookup{\,} : \classes\times\mathit{MethodName} \rightharpoonup \classes$
formalizes \emph{Java's dynamic method lookup}, \textit{i.e.} the runtime process of determining
the class containing the implementation of a method on the basis of the class of the receiver object:
$ \lookup{\kappa, m}  \eqdef   \textit{min}(\uparrow\!\kappa.m)$ if  $\uparrow\!\kappa.m\neq\emptyset$ and is undefined otherwise, 
where $\uparrow\!\kappa.m  \eqdef  \{\kappa'\in\,\uparrow\!\kappa \; \mid \;  m \textrm{ is implemented in } \kappa'\}$
%
is a finite chain since $\uparrow\!\kappa.m\subseteq\, \uparrow\!\kappa$.

The set of \emph{expressions} $\expressions$, ranged over by
$E$, and the set of \emph{commands} $\commands$, ranged over by $C$, are defined
as follows. \emph{Method bodies}, ranged over by 
$B$, are \void-terminated commands.
\[
\begin{array}{rcl}
E & {\Bdf} & x \Bor E.f \Bor \objinit{\kappa}{f_1 = E_1,\ldots,f_n = E_n} \\[2pt]
C & {\Bdf} & \dec{x = E} \Bor x \assign E \Bor x.f \assign E \Bor C;C \Bor 
\! \void \! 
 \Bor \! E.m(\,) \! \Bor   \\
 & &  {\fork{E.m(\,)}} \! \Bor \! \synch{E}{C}  \!\Bor\! \lock{l} \!\Bor\! \unlock{l} \\[2pt]
B & {\Bdf} & \mbox{\void} \Bor C;\void
\end{array}
\]
Constructs of our language are simplified versions of those of Java.
For instance, loops must be implemented through recursion. We assume that the compiler
ensures some standard syntactical properties, such as
the same  variable cannot be declared twice in a method, and
the only free variable in a method's body is \<this>.
These simplifying assumptions can be relaxed without affecting our results. 

Expressions are variables, field accesses, and a construct for object creation,
$\objinit{\kappa}{f_1=E_1,\ldots,f_n=E_n}$, that creates an object of class $\kappa$ and initializes
each field $f_i$ to the value of $E_i$.
Command $\dec{\!}$ declares a local variable.
The declaration of a local variable in the body $B$ of a method $m$ must
introduce a \emph{fresh} variable never declared before in $B$,
whose lifespan starts from there and reaches the end of $B$.
The commands for variable/field assignment, sequential composition, 
and termination are standard. 
Method call
$E.m(\,)$ looks up and runs method $m$ on the runtime value of $E$.
Command $\fork{E.m(\,)}$ does the same asynchronously, on a new thread.
Command $\synch{E}{C}$ is like Java's
\<synchronized>: the command $C$ can be executed only when 
the current thread holds the lock on the value of $E$. 
$\lock{l}$  and $\unlock{l}$ cannot be used by the programmer: 
our semantics introduces them in order to implement object synchronization.

The set of classes is 
\(\classes\eqdef\{\kappa :\mathit{MethodNames}\rightharpoonup B \mid\dom{\kappa}\textrm{ is finite}\}  . \)
The binding of fields to their defining class is not relevant in our formalization.
Given a class $\kappa$ and a method name $m$, if $\kappa(m) = B$ then $\kappa$
implements $m$ with body $B$\@.
For simplicity, $\this$ is the only free variable in $B$ and
methods have no formal parameters and/or return value.
A \emph{program} is a finite set of classes and includes
a distinguished class $\initk$ that only defines a
method $\mathit{main}$ where the program starts:
$\initk\; \mathit{\eqdef}\; \{\methodef{\mathit{main}}{B_{\mathit{main}}}\}$.
\label{sec:ex:simple}
\begin{figure}[t]
\begin{tabular}{c|cc}
{\scriptsize\qquad\begin{lstlisting}
public class K {
  private K1 x = new K1(); 
  private K2 y = new K2();
  public void m() {
    K1 z = x;
    K2 w = new Object();
    synchronized (z) {
      y = z.f;
      w = y;
    }
    w.g = new Object();
  }
}

class K1 {
  K2 f = new K2();
}

class K2 {
  Object g = new Object();
}
\end{lstlisting}\qquad} 
&  
\qquad
&
\begin{tabular}{rr|c|c|}
	  &      & by-name & by-value \\
\hline    
field & \<x> & --                         & \<@GuardedBy(itself)> \\
field & \<y> & \<@GuardedBy(this.x)> & --                    \\
variable   & \<z> & \<@GuardedBy(itself)>      & \<@GuardedBy(itself)> \\
variable   & \<w> & --                         & --  \\
\hline
\end{tabular}
\end{tabular}
\caption{Running example.}
\label{fig:semantics_differences}
\end{figure}
\begin{sloppy}
\begin{example}\label{ex:simple}
\vspace*{-2mm}
Fig.~\ref{fig:semantics_differences} gives our \emph{running example\/} in Java.
In our core language, the body of method \<m> is translated as follows:
$B_{\<m>}=\strut$\<decl z = this.x;
decl w = Object$\langle\rangle$;
sync(z) \{ this.y {:=} z.f; w {:=} this.y \};
w.g {:=} Object$\langle\rangle$;
skip>,
with classes
$\<K> \mathrel{\eqdef} \{\methodef{\<m>}{B_{\<m>}}\}$,
$\<K1> \mathrel{\eqdef} \phi$, $\<K2> \mathrel{\eqdef} \phi$, and
$\<Object> \mathrel{\eqdef} \phi$.
\end{example}
\end{sloppy}



%% file: semantics.tex
\subsection{Semantic Domains}\label{subsec:domains}
A running program consists of a pool of threads that share a memory.
Initially, a single thread runs the main method. The $\fork{E.m(\,)}$ command
adds a new thread 
to the existing ones. 
Each thread has an activation stack $S$ and
a set $\Loc$ of locations that it currently locks.
The activation stack $S$ is a stack of activation records $R$ of methods.
Each $R$ consists of the identifier $\kappa.m$ of the method, the command $C$ to be executed
when $R$ will be on top of the stack (\emph{continuation}),
and the \emph{environment} or binding $\sigma$ that provides values to the variables in scope in $R$\@.
For simplicity, we only have classes and no primitive types, so
the only possible \emph{values} are locations. Formally, 
\(\states \; \eqdef \; \{ \sigma : \varset \rightharpoonup \locations  \mid \dom{\sigma}\textrm{ is finite}\}.
\)
\begin{definition}
\emph{Activation records}, ranged over by $R$,
 \emph{activation stacks}, ranged over by $S$, and
\emph{thread pools}, ranged over by $T$,  are defined as follows:
\vspace*{-.5mm}
\begin{center}
\(
\begin{array}{rcll}
R & \Bdf &  \record{\kappa.m}{C}{\sigma} & \text{(activation record for $\kappa.m$)}\\
S & \Bdf & \emptystack \Bor R\stsep S & \text{(activation stack, possibly empty)}\\
T & \Bdf & \thread{S}{\lockset} \Bor T\parallel T\quad & \text{(thread pool)}.
\end{array}
\)
\vspace*{-.5mm}
\end{center}
The number of threads in $T$ is written as $\#T$.
\end{definition} 

An object $o$ is a triple  
containing 
the object's class, an environment
binding its fields to their corresponding values, and  
a lock, \textit{i.e.}, an integer
counter incremented whenever a thread locks the object (locks are re-entrant). A \emph{memory} $\mu$ 
maps a finite set of already allocated
 memory locations into \emph{objects}.
\begin{definition}
\emph{Objects} and \emph{memories} are defined as 
$
\objects \eqdef \classes \times \states \times\mathbb{N}$ and 
$\memory \eqdef \{\mu : \locations \rightharpoonup  \objects  \mid \dom{\mu}\textrm{ is finite}\}$, 
with selectors
$\cls{o}\eqdef \kappa$, $\state{o}\eqdef \sigma$ and $\lockson{o}\eqdef n$, 
for every $o = \langle \kappa, \sigma, n\rangle \in \objects$.
We  also define $\extends{o}{f}{l}\eqdef\triplet{\kappa}{\extends{\sigma}{f}{l}}{n}$ and
$\addlock{o} \eqdef \triplet{\kappa}{\sigma}{n{+}1}$ and $\remlock{o} \eqdef \triplet{\kappa}{\sigma}{\max(0, n{-}1)}$.
\end{definition}
For simplicity, we do not model delayed publication of field updates, allowed in 
the Java memory model, as that is not relevant for our semantics and results.
Our goal is to identify expressions definitely locked at selected program points
and locking operations are immediately published in the Java memory model.
Hence, our memory model is a deterministic map shared by all threads.

The \emph{evaluation of an expression} $E$ in an environment $\sigma$ and in
a memory $\mu$, written $\eval{E}{\mu}{\sigma}$, yields a pair $\langle l,\mu'\rangle$, where $l$ is a location (the runtime
value of $E$) and $\mu'$ is the memory resulting 
after the evaluation of $E$.
Given a pair $\langle l,\mu \rangle$ we use
selectors $\loc{\langle l,\mu \rangle}=l$ and 
$\mem{\langle l,\mu \rangle}=\mu$.
\begin{definition}[Evaluation of Expressions]
\label{def:evaluation_expressions}
The evaluation function  has the type 
$
 \eval{\ }{}{} : (\expressions \times \states \times \memory) \rightharpoonup (\locations \times \memory)
$
and is defined as: \vspace*{-1ex}
\begin{align*}
  \eval{x}{\mu}{\sigma}  \deff \anglepair{\sigma(x)}{\mu}\qquad\Q
  \eval{E.f}{\mu}{\sigma} & \deff \anglepair{\state{\mu'(l)}(f)}{\mu'},\text{ where}\eval{E}{\mu}{\sigma}{=}\anglepair{l}{\mu'} \\
  \eval{\objinit{\kappa}{f_1{=}E_1,..,f_n{=}E_n}}{\mu}{\sigma} & \deff \anglepair{l}{\extends{\mu_n}{l}{\triplet{\kappa}{\sigma'}{0}}}\text{, where}
\end{align*}
\vspace*{-5ex}
\begin{align*}
  (1)\, & \mu_0 = \mu\text{ and }\anglepair{l_i}{\mu_i}=\eval{E_i}{\mu_{i-1}}{\sigma},\textrm{ for }i\in[1.. n]\\
  (2)\, & \text{$l$ is \emph{fresh} in $\mu_n$, that is $\mu_{n}(l)\uparrow$} \\
  (3)\, & \sigma'\in\states\text{ is such that\ }
        \sigma'(f_i)=l_i\textrm{ for }i\in[1..n]\text{, while }\sigma'(y){\uparrow}\text{ elsewhere}. \\[-7mm]
\end{align*} 
We assume that $\eval{\ }{}{}$ is undefined if any of the function applications is undefined.
\end{definition}
In the evaluation of the object creation expression, a fresh location $l$ is allocated and
bound to an unlocked object whose environment $\sigma'$ binds its fields to the values of the
corresponding initialization expressions.
\subsection{Structural Operational Semantics}\label{subsec:sos}

Our operational semantics is given in terms of a \emph{reduction relation}
on \emph{configurations} of the form $\anglepair{T}{\mu}$, where $T$ is a pool of threads
and $\mu$ is a memory that models the heap of the system. 
We write $\anglepair{T}{\mu}  \reduces{n} \anglepair{T'}{\mu'}$ for representing 
an execution step 
in which $n\ge 1$ denotes the position of the thread in $T$ that fires the transition,
starting from the leftmost in the pool $T$ (thread $1$).
We write $\reduces{}$ instead of $\reduces{n}$ to abstract on the running thread;
$\reduces{}^\ast$ denotes the reflexive and transitive closure of $\reduces{}$. 
We first introduce
reduction rules
where the activation stack consists of a single activation record,
then lift to the general case.

\begin{table}[t]
\(
\begin{array}{l}
  \infer[\textrm{[decl]}]{\anglepair{\thread{\record{\kappa.m}{\dec{x=E}}{\sigma}}{\lockset}}{\mu} 
             \reduces{1} \anglepair{\thread{\record{\kappa.m}{\void}{\sigma'}}{\lockset}}{\mu'}}
        {\eval{E}{\mu}{\sigma} = \anglepair{l}{\mu'} \quad \sigma(x)\!\uparrow \quad \sigma'\eqdef \extends{\sigma}{x}{l}}
\\[13pt]
 \infer[\textrm{[var-ass]}]{\anglepair{\thread{\record{\kappa.m}{x\assign E}{\sigma}}{\lockset}}{\mu} 
        	\reduces{1} \anglepair{\thread{\record{\kappa.m}{\void}{\sigma'}}{\lockset}}{\mu'}}
        {\eval{E}{\mu}{\sigma} = \anglepair{l}{\mu'} \quad \sigma(x)\!\downarrow \quad \sigma'\eqdef \extends{\sigma}{x}{l}}
\\[13pt]
  \infer[\textrm{[field-ass]}]{\anglepair{\thread{\record{\kappa.m}{x.f\assign E}{\sigma}}{\lockset}}{\mu}
             \reduces{1} \anglepair{\thread{\record{\kappa.m}{\void}{\sigma}}{\lockset}}{\mu''}}
        {\eval{E}{\mu}{\sigma} = \anglepair{l}{\mu'} \quad o = \mu(\sigma(x)) \quad o'\eqdef\extends{o}{f}{l} \quad \mu''\eqdef\extends{\mu'}{\sigma(x)}{o'}}
\\[13pt]
  \infer[\textrm{[seq]}]{\anglepair{\thread{\record{\kappa.m}{C_1; C_2}{\sigma}}{\lockset}}{\mu}
              \reduces{1} \anglepair{\thread{\record{\kappa.m}{C'_1; C_2}{\sigma'}}{\lockset'}}{\mu'}}
        {\anglepair{\thread{\record{\kappa.m}{C_1}{\sigma}}{\lockset}}{\mu}
              \reduces{1} \anglepair{\thread{\record{\kappa.m}{C'_1}{\sigma'}}{\lockset'}}{\mu'} \q\, \textrm{$C_1\neq E.p()$ \q\, $C_1\neq \fork{E.p()}$}}
\\[13pt]
  \infer[\textrm{[seq-skip]}]{\anglepair{\thread{\record{\kappa.m}{\void; C}{\sigma}}{\lockset}}{\mu}
               \reduces{1} \anglepair{\thread{\record{\kappa.m}{C}{\sigma}}{\lockset}}{\mu}}
        {-}
\\[13pt]
  \infer[\textrm{[invoc]}]{\anglepair{\thread{\record{\kappa.m}{E.p(\;); C}{\sigma}\stsep S}{\lockset}}{\mu}
               \reduces{1}{\anglepair{\thread{\record{{\kappa'}{.}p}{B}{\{\this\mapsto l\}} \stsep \record{\kappa.m}{C}{\sigma} \stsep S}{\lockset}}{\mu'}}}
        {\eval{E}{\mu}{\sigma} = \anglepair{l}{\mu'} \quad  \kappa' = \lookup{\cls{\mu'(l)},p} \quad \kappa'(p)=B}
\end{array}
\)
\caption{Structural operational semantics for sequential commands.}
\label{tab:SOS-sequential}
\end{table}
Table~\ref{tab:SOS-sequential} deals with 
 sequential commands. In rule $\textrm{[decl]}$ an 
undefined variable $x$ is declared. 
Rules $\textrm{[var-ass]}$ and $\textrm{[field-ass]}$ formalize variable and field assignment, respectively. 
Rule $\textrm{[seq]}$ assumes that the 
first command 
is not of the form $E.p(\;)$ or 
$\fork{E.p(\;)}$;
these two cases are treated separately. In rule $\textrm{[invoc]}$ 
the receiver $E$ is evaluated and the method implementation is looked up
from the dynamic class of the receiver. The body of the method is put on top 
of the activation stack and is executed
from an initial state where only variable $\this$ is in scope, bound to the receiver. 
Unlike previous rules, this rule deals with the whole activation stack
rather than assuming only a single activation record.

\begin{table}[t]
\(
\begin{array}{l}
  \infer[\textrm{[spawn]}]{\anglepair{\thread{\record{\kappa.m}{\fork{E.p(\;)}; C}{\sigma}\stsep S}{\lockset}}{\mu}
               \reduces{1}{\anglepair{\thread{\record{{\kappa'}{.}p}{B}{\{\this\mapsto l\}} \stsep \epsilon}{\emptyset} \parallel \thread{\record{\kappa.m}{C}{\sigma} \stsep S}{\lockset}}{\mu'}}}
        {\eval{E}{\mu}{\sigma} = \anglepair{l}{\mu'} \quad \kappa' = \lookup{\cls{\mu'(l)},p} \quad \kappa'(p)=B}
\\[12pt]
  \infer[\textrm{[sync]}]{\anglepair{\thread{\record{\kappa.m}{\synch{E}{C}}{\sigma}}{\lockset}}{\mu}
               \reduces{1}\anglepair{\thread{\record{\kappa.m}{\lock{l};C;\unlock{l}}{\sigma}}{\lockset}}{\mu'}}
         {\eval{E}{\mu}{\sigma} = \langle l,\mu'\rangle}
\\[12pt]
  \infer[\textrm{[acquire-lock]}]{\anglepair{\thread{\record{\kappa.m}{\lock{l}}{\sigma}}{\lockset}}{\mu}
               \reduces{1}\anglepair{\thread{\record{\kappa.m}{\void}{\sigma}}{\lockset'}}{\mu'}}
         { \lockson{\mu(l)} = 0  \quad  \lockset'\eqdef \lockset \cup \{ l \} \quad \mu' \eqdef \extends{\mu}{l}{\addlock{\mu(l)}}}
\\[12pt]
  \infer[\textrm{[reentrant-lock]}]{\anglepair{\thread{\record{\kappa.m}{\lock{l}}{\sigma}}{\lockset}}{\mu}
               \reduces{1}\anglepair{\thread{\record{\kappa.m}{\void}{\sigma}}{\lockset}}{\mu'}}
         { l\in\lockset \quad \mu' \eqdef \extends{\mu}{l}{\addlock{\mu(l)}}}
\\[12pt]
  \infer[\textrm{[decrease-lock]}]{\anglepair{\thread{\record{\kappa.m}{\unlock{l}}{\sigma}}{\lockset}}{\mu}
               \reduces{1}\anglepair{\thread{\record{\kappa.m}{\void}{\sigma}}{\lockset}}{\mu'}}
         {
\lockson{\mu(l)}>1 \quad \mu' \eqdef \extends{\mu}{l}{\remlock{\mu(l)}} }
\\[12pt]
  \infer[\textrm{[release-lock]}]{\anglepair{\thread{\record{\kappa.m}{\unlock{l}}{\sigma}}{\lockset}}{\mu}
               \reduces{1}\anglepair{\thread{\record{\kappa.m}{\void}{\sigma}}{\lockset'}}{\mu'}}
         {
\lockson{\mu(l)} = 1 \quad \lockset'\eqdef\lockset\setminus\{l\} \quad \mu' \eqdef \extends{\mu}{l}{\remlock{\mu(l)}}}
\end{array}
\)
\caption{Structural operational semantics for concurrency and synchronization.}
\label{tab:SOS-concurrent}
\end{table}
\begin{table}[t!]
\(
\begin{array}{ll}
  \infer[\textrm{[push]}]{\anglepair{\thread{R{::} S}{\lockset}}{\mu}
               \reduces{1}\anglepair{\thread{R'{::} S}{\lockset'}}{\mu'} }
        {\anglepair{\thread{R}{\lockset}}{\mu}
               \reduces{1}\anglepair{\thread{R'}{\lockset'}}{\mu'} }        
&
  \infer[\textrm{[pop]}]{\anglepair{\thread{\record{\kappa.m}{\void}{\sigma}{::} S}{\lockset}}{\mu}
              \reduces{1}\anglepair{\thread{S}{\lockset}}{\mu}}
        {-}
\\[6pt]
  \infer[\textrm{[par-l]}]{\anglepair{T_1 \parallel T_2}{\mu} 
              \reduces{n}\anglepair{T_1' \parallel T_2}{\mu'}}
        {\anglepair{T_1}{\mu}\reduces{n}\anglepair{T_1'}{\mu'}}
 &
  \infer[\textrm{[par-r]}]{\anglepair{T_1 \parallel T_2}{\mu} 
              \reduces{\#T_1 + n}\anglepair{T_1 \parallel T_2'}{\mu'}}
        {\anglepair{T_2}{\mu}\reduces{n}\anglepair{T_2'}{\mu'}}
\\[6pt]
  \infer[\textrm{[end-l]}]{\anglepair{\thread{\emptystack}{\lockset} \parallel T}{\mu}
              \reduces{1}\anglepair{T}{\mu}}
        {-}
&    
  \infer[\textrm{[end-r]}]{\anglepair{T \parallel \thread{\emptystack}{\lockset}}{\mu}
               \reduces{\#T + 1}\anglepair{T}{\mu}}
        {-}
\end{array}
\)
\caption{Structural operational semantics: structural rules.}
\label{tab:SOS-structural}
\end{table}
Table~\ref{tab:SOS-concurrent} focuses on concurrency and synchronization. 
The spawn of a new method is similar to a method call, but the
method body runs in its own new thread with an initially empty set of locked locations.  In rule $\textrm{[sync]}$ the location $l$ associated to the guard $E$ is computed; the computation can proceed only if a lock action is possible 
on $l$. 
The lock will be released only at the 
end of the critical section $C$. Rule $\textrm{[acquire-lock]}$ models 
the  entering of the monitor of an unlocked object. 
%
%
Rule $\textrm{[reentrant-lock]}$ models Java's \emph{lock reentrancy}. 
Rule $\textrm{[decrease-lock]}$ decreases the lock counter of an object that still remains locked, as it was
locked more than once. 
When the lock counter reaches $0$, rule $\textrm{[release-lock]}$ can fire to release the lock of the object. 

In Table~\ref{tab:SOS-structural}, rule $\textrm{[push]}$
lifts the execution of an activation record to that of a 
stack of activation records. The remaining structural 
rules are straightforward. 
\begin{definition}[Operational Semantics of a Program]
The initial configuration of a program is $\anglepair{T_0}{\mu_0}$ where
$T_0 \eqdef \thread{\record{\initk.\mathit{main}}{B_{\mathit{main}}}{\{\this\mapsto\initl\}}}{\emptyset}$
and $\mu_0 \eqdef \mapdef{\initl\mapsto\triplet{\initk}{\phi}{0}}$.
The \emph{operational semantics} of a program is the set of traces of the form
$
  \anglepair{T_0}{\mu_0}\ \reduces{}^\ast\ \anglepair{T}{\mu}.
$
\end{definition}

\begin{example}\label{ex:simple_semantics}
The implementation in Ex.~\ref{ex:simple}, 
becomes a program by defining $B_{\mathit{main}}$ as:
${\small
\objinit{\<K>}{\,
\<x> = \objinit{\<K1>}{\<f> = \objinit{\<K2>}{\<g> = \objinit{\<Object>}{}}},\,
\<y> = \objinit{\<K2>}{\<g> = \objinit{\<Object>}}
}.\<m>();\void.}
$
The operational semantics builds the following maximal trace
from $\anglepair{T_0}{\mu_0}$ that, for convenience, we divide in eight macro-steps:\\[2pt]
\(
{\small
\begin{array}{lcl}
1.\; & \reduces{}^{\ast} & \anglepair{\thread{\record{\<K>.\<m>}{\dec{\<z> = \this.\<x>};\ldots}{\sigma_1}\stsep\record{\initk.\mathit{main}}{\void}{\{\this\mapsto\initl\}}}{\emptyset}}{\mu_1} \\
& & 
\textrm{with } 
\mu_1 \eqdef \mu_0[l\mapsto o, l_1\mapsto o_1, l_2\mapsto o_2, l_3\mapsto o_3, l_4\mapsto o_4, l_5\mapsto o_4]; \\
& & 
o \eqdef \triplet{\<K>}{\{\<x>\mapsto l_1, \<y>\mapsto l_2\}}{0}; \;
o_1 \eqdef \triplet{\<K1>}{\{\<f>\mapsto l_3\}}{0}; \;
o_2 \eqdef \triplet{\<K2>}{\{\<g>\mapsto l_4\}}{0};\\
& & 
o_3 \eqdef \triplet{\<K2>}{\{\<g>\mapsto l_5\}}{0}; \;
o_4 \eqdef \triplet{\<Object>}{\phi}{0}; \;
\sigma_1 \eqdef \{\this\mapsto l\}
\\
2. & \reduces{}^{\ast} & \anglepair{\thread{\record{\<K>.\<m>}{\dec{\<w> = \objinit{\<Object>}{}};\ldots}{\sigma_2}\stsep\ldots}{\emptyset}}{\mu_1}
\textrm{ with }\sigma_2 \eqdef \extends{\sigma_1}{\<z>}{l_1} \\
3. & \reduces{}^{\ast} & \anglepair{\thread{\record{\<K>.\<m>}{\synch{\<z>}{\ldots};\ldots}{\sigma_3}\stsep\ldots}{\emptyset}}{\mu_2}
\textrm{ with } 
\mu_2 \eqdef \mu_1[l_6 \mapsto o_4 ]; \;
\sigma_3 \eqdef \extends{\sigma_2}{\<w>}{l_6}
\\
4. & \reduces{}^{\ast} & \anglepair{\thread{\record{\<K>.\<m>}{\this.\<y> \assign \<z>.\<f>;\ldots;\unlock{l_1};\ldots}{\sigma_3}\stsep\ldots}{\{l_1\}}}{\mu_3} \\
& &  
\textrm{with } 
\mu_3 \eqdef \mu_2[l_1 \mapsto \addlock{o_1}]
\\
5. & \reduces{}^{\ast} & \anglepair{\thread{\record{\<K>.\<m>}{\<w> \assign \this.\<y>;\unlock{l_1};\ldots}{\sigma_3}\stsep\ldots}{\{l_1\}}}{\mu_4} \\
& & 
\textrm{with } 
\mu_4 \eqdef \mu_3[l \mapsto o']; \;
o' \eqdef \triplet{\<K>}{\{\<x>\mapsto l_1, \<y>\mapsto l_3\}}{0} 
\\
6. & \reduces{}^{\ast} & \anglepair{\thread{\record{\<K>.\<m>}{\unlock{l_1};\<w>.\<g> :=  \objinit{\<Object>}{};\void}{\sigma_4}\stsep\ldots}{\{l_1\}}}{\mu_4} \textrm{ with }\sigma_4 \eqdef \extends{\sigma_3}{\<w>}{l_3}
\\
7. & \reduces{}^{\ast} & \anglepair{\thread{\record{\<K>.\<m>}{\<w>.\<g> :=  \objinit{\<Object>}{};\void}{\sigma_4}\stsep\ldots}{\emptyset}}{\mu_5} \textrm{ with } \mu_5 \eqdef \mu_4[l_1 \mapsto {o_1}] \\
8. & \reduces{}^{\ast} & \anglepair{\thread{\record{\initk.\mathit{main}}{\void}{\{\this\mapsto\initl\}}}{\emptyset}}{\mu_6} \textrm{ with }\mu_6 \eqdef \mu_5[l_5 \mapsto o_4].
\end{array}
}
\)
\end{example}

\label{sec:ex:simple_semantics}
Our semantics lets us formalize some properties
on the soundness of the locking mechanism, that we report
in Appendix~\ref{sec:properties}.
Here we just report a key property used in our proofs,
that states
that two threads never lock the same location (\emph{i.e.}, object)
at the same time. It is proved by induction on the length of the trace.
\looseness=-1

\begin{proposition}
\label{prop:lock-threads}
Let $\anglepair{T_0}{\mu_0}\trans{}^*\anglepair{\thread{S_1}{\lockset_1}\parallel
\ldots\parallel \thread{S_n}{\lockset_n}\,}{\mu}$ be an arbitrary trace.
For any $i,j\in\{1\ldots n\}$, $i\neq j$ entails $\lockset_i\cap\lockset_j = \emptyset$.
\end{proposition}

%% file: guardedBy.tex
\section{Two Semantics for \<@GuardedBy> Annotations}\label{sec:guardedBy}
This section gives two distinct formalizations
for locking specifications
 of the form \<@GuardedBy(\|E|) \emph{Type} x>,
where $E$ is any expression allowed by 
the language, possibly using a special
variable \<itself> that stands for the protected entity.
\subsection{Name-Protection Semantics}\label{subsec:byname}
In a \emph{name-protection} interpretation,
a thread must hold the lock on the value of the guard expression
whenever it \emph{accesses} (reads or writes) the \emph{name} of the 
guarded variable/field. 
Def.~\ref{def:access} formalizes the notion of
\emph{accessing an expression} when a given command is executed.
For our purposes, it is enough to consider a single execution step; thus
the accesses in $C_1;C_2$ are only those in $C_1$.
When an object is created, only its creating thread can access it.
Thus field initialization cannot originate data races and is
not considered as an access.
The access refers to the value of the expression, not
to its lock counter, hence $\synch{E}{C}$ does not access $E$\@.
For accesses to a field $f$, Def.~\ref{def:access}
keeps the exact expression used for the container of $f$, that
will be used in Def.~\ref{def:guardedByFields} for the contextualization of \<this>.
\begin{definition}[Expressions Accessed in a Single Reduction Step]
\label{def:access}
The set of expressions accessed 
in a single execution step is defined as follows:\\[2pt]
\noindent
{\em 
\(
\begin{array}{@{\hspace*{2mm}}rcl@{\hspace*{6mm}}rcl}
\acc{x} & \eqdef & \{x\} &\acc{E.f} & \eqdef & \acc{E} \cup \{E.f\} \\ 
\multicolumn{4}{r}{
\acc{\objinit{\kappa}{f_1{=} E_1,\ldots,f_n {=} E_n}}} & \eqdef &
\bigcup_{i=1}^{n}\acc{E_i}
\\
  \acc{\dec{x = E}} &\eqdef& \acc{E} 
&
  \acc{x \assign E} &\eqdef& \acc{x} \cup \acc{E} \\
 \acc{C_1;C_2} &\eqdef& \acc{C_1}
&
 \acc{x.f \assign E} &\eqdef& \acc{x.f} \cup \acc{E} 
  \\

 \acc{E.m(\,)} &\eqdef& \acc{E} 
&
  \acc{\fork{E.m(\,)}} &\eqdef& \acc{E}  \\

 \acc{\lock{l}} & \eqdef & \emptyset
&
\acc{\unlock{l}} & \eqdef & \emptyset
\\
 \acc{\synch{E.f}{C}} &\eqdef& \acc{E}
&
\acc{\synch{x}{C}} &\eqdef& \emptyset \; \eqdef \;  \acc{\void}\; 
\\
\multicolumn{6}{c}{\acc{\synch{\objinit{\kappa}{f_1 = E_1,\ldots,f_n = E_n}}{C}} \, \eqdef \, \acc{\objinit{\kappa}{f_1 {=} E_1,\ldots,f_n {=}E_n}}.}
\end{array}
\)
}

\noindent
We say that a command $C$ \emph{accesses} a variable $x$  if and only if $x\in\acc{C}$;
we say that $C$ \emph{accesses} a field $f$ if and only if $E.f\in\acc{C}$, for some expression $E$.
\end{definition}

We now define \<@GuardedBy> for local variables
(Def.~\ref{def:guardedByVars}) and for fields (Def.~\ref{def:guardedByFields}).
In Sec.~\ref{sec:examples} we have already discussed the reasons for  
using the special variable \<itself> in the guard expressions 
when working with a value-protection semantics. In the name-protection semantics,   \<itself> denotes just an alias of the accessed name: 
\<@GuardedBy(itself) \emph{Type} x> 
is the same as \<@GuardedBy(x) \emph{Type} x>.
\begin{definition}[\<@GuardedBy> for Local Variables]
\label{def:guardedByVars}
A local variable $x$ of a method $\kappa.m$ in a program is 
\emph{name protected} by {\em $\guardedby{E}$}
if and only if for every derivation
$
  \anglepair{T_0}{\mu_0}\ \trans{}^*\ \anglepair{T}{\mu}\reduces{n}\cdots
$
in which the $n$-th thread in $T$ is $\thread{\record{\kappa.m}{C}{\sigma}\stsep S}{\lockset}$,
whenever $C$ accesses $x$ we have $\loc{\eval{E}{\mu}{\sigma[\mathtt{itself}\mapsto\sigma(x)]}} \in \lockset$.
\end{definition}
\begin{example}
In Ex.~\ref{ex:simple_semantics}, 
variable \<z> of $\<K.m>$ is name protected by \<@GuardedBy(this.x)> since
the name \<z> is accessed at the macro-step $5$ only, where
$\eval{\<this.x>}{\mu_3}{\sigma_3[\mathtt{itself}\mapsto l_1]}=\anglepair{l_1}{\mu_3}$; during those reductions, the current
thread holds the lock on the object bound to $l_1$. According to
Def.~\ref{def:access}, macro-steps $2$ and $3$ do not contain accesses since 
they are declarations; macro-step $4$ does not access \<z> since it is a synchronization.
\end{example}
\begin{definition}[\<@GuardedBy> for Fields]
\label{def:guardedByFields}
A field $f$ in a program is \emph{name protected} by {\em $\guardedby{E}$}
if and only if for every trace 
$
  \anglepair{T_0}{\mu_0}\ \trans{}^*\ \anglepair{T}{\mu}\reduces{n}\ldots
$
in which the $n$-th thread in $T$ is $\thread{\record{\kappa.m}{C}{\sigma}\stsep S}{\lockset}$, whenever $C$ accesses $f$,  i.e.\
 $E'.f\in\acc{C}$, for some $E'$, with  $\eval{E'}{\mu}{\sigma}=\anglepair{l'}{\mu'}$ and $l''=\state{\mu'(l')}f$, we have 
$\loc{\eval{E}{\mu'}{\sigma[\mathtt{this}\mapsto l',\mathtt{itself}\mapsto {l''}]}}\in\lockset$.
\end{definition}
Notice that the guard expression $E$ is evaluated in a memory $\mu'$ obtained 
by the evaluation of the container of $f$, that is $E'$, and in an environment where  
the special variable \<this> is bound to $l'$, \emph{i.e.} the evaluation of the container of $f$.
\begin{remark}
\label{rem:program-points}
Def.~\ref{def:guardedByVars} and~\ref{def:guardedByFields}
evaluate the guard $E$ at those program points where
$x$ is accessed, in order to verify that its lock is held by the current 
thread. Hence $E$ can only refer to \<itself> and variables in scope at those
points, and for its evaluation we must use the current
environment $\sigma$. A similar observation holds
for the corresponding definitions for the value-protection semantics
in next section.
\end{remark}
\begin{example}
In Ex.~\ref{ex:simple_semantics}, field \<y> is name protected by \<@GuardedBy(this.x)>.
It is accessed at macro-step $5$, 
where $\eval{\<this.x>}{\mu_3}{\sigma_3[\mathtt{this}\mapsto l, \,
\mathtt{itself}\mapsto l_1]} = \anglepair{l_1}{\mu_3}$, and
at macro-step $6$, where $\eval{\<this.x>}{\mu_4}{\sigma_4[\mathtt{this}\mapsto l,\,
\mathtt{itself}\mapsto l_1]} = \anglepair{l_1}{\mu_4}$. In both cases,
the active and only thread holds the lock on the object bound to $l_1$.
\end{example}

\subsection{Value-Protection Semantics}
\label{sec:gb-value}
An alternative semantics for \<@GuardedBy> protects the values held in variables or
fields rather than their name.
In this \emph{value-protection} semantics, a variable $x$ is $\guardedby{E}$ if
wherever a thread dereferences a location $l$ eventually bound to $x$,
it holds the lock on the object obtained by evaluating $E$ at that point.
In object-oriented parlance, \emph{dereferencing a location $l$} means accessing the object stored at
$l$ in order to read or write a field.
In Java, accesses to the lock counter are synchronized at a low level and
the class tag is immutable, hence their accesses cannot give rise to data races and
are not relevant here. Dereferences (Def.~\ref{def:deref})
are very different from  accesses (Def.~\ref{def:access}).
For instance, statement \<v.f := w.g.h> accesses expressions
\<v>, \<v.f>, \<w>, \<w.g> and \<w.g.h>
but dereferences only the locations held in \<v>, \<w> and \<w.g>: locations bound to \<v.f> and \<w.g.h> are left untouched.
Def.~\ref{def:deref}
formalizes the set of locations dereferenced by an
expression or command to access some field and keeps track
of the fact that the access is for reading ($\rightarrow$)
or writing ($\leftarrow$) the field. Hence dereference tokens are $l.f\!\!\leftarrow$
or $l.f\!\!\rightarrow$, where $l$ is a location
and $f$ is the name of the field that is accessed in the object held in $l$.
\begin{definition}[Dereferenced Locations]\label{def:deref}
Given a memory $\mu$ and an environment $\sigma$, 
the dereferences in a single reduction are defined as follows:\\[1mm]
\noindent
{\em 
\(\begin{array}{rclcrcl}
\multicolumn{7}{c}{
\rfrc{x}{\mu}{\sigma}
 \eqdef \emptyset \Q\Q
\qquad\rfrc{E.f}{\mu}{\sigma}   \eqdef  \left\{\loc{\eval{E}{\mu}{\sigma}}.\mathord{f\!\!\rightarrow}\right\} {\cup} \rfrc{E}{\mu}{\sigma}}\\ 

\multicolumn{7}{c}{\rfrc{\objinit{\kappa}{f_1 = E_1,\ldots,f_n = E_n}}{\mu}{\sigma}   \eqdef  \bigcup^{n}_{i = 1}\rfrc{E_i}{\mu}{\sigma}
}
\\
 \rfrc{\dec{x = E}}{\mu}{\sigma} & \eqdef & \rfrc{E}{\mu}{\sigma} &&
 \rfrc{x \assign E}{\mu}{\sigma} & \eqdef &\rfrc{E}{\mu}{\sigma}\\
\rfrc{\synch{E}{C}}{\mu}{\sigma} & \eqdef &\rfrc{E}{\mu}{\sigma}
&&
 \rfrc{C_1;C_2}{\mu}{\sigma} & \eqdef &\rfrc{C_1}{\mu}{\sigma}
 \\
\rfrc{\lock{l}}{\mu}{\sigma} & \eqdef &\emptyset 
  &&
\rfrc{x.f \assign E}{\mu}{\sigma} &\eqdef& 
{\small \mbox{$\{\sigma(x).\mathord{f\!\!\leftarrow}\}{\cup}\rfrc{E}{\mu}{\sigma}$}}
\\
 \rfrc{\unlock{l}}{\mu}{\sigma} &\eqdef&\emptyset
&&
 \rfrc{\void}{\mu}{\sigma}& \eqdef &\emptyset\\
\multicolumn{7}{c}{ \rfrc{E.m(\,)}{\mu}{\sigma} \q \eqdef \q \rfrc{\fork{E.m(\,)}}{\mu}{\sigma} \q \eqdef\q
   \rfrc{E}{\mu}{\sigma} \enspace .  }
\end{array}
\)
}\\
Its projection on locations is
$\mathit{derefloc}(C)_\sigma^\mu=\{l\mid\text{there is }f\text{ such that }l.\mathord{f\!\!\leftarrow}\in\rfrc{C}{\mu}{\sigma}
\text{ or }l.\mathord{f\!\!\rightarrow}\in\rfrc{C}{\mu}{\sigma}\}$.
\end{definition}


Def.~\ref{def:gardedByValueVariables} fixes an arbitrary execution trace $t$
and collects the set $\mathcal{L}$ of locations that have ever been 
bound to $x$ in $t$. Then, it requires that whenever a thread 
dereferences one of those locations, that thread must hold
the lock on the object obtained by evaluating the guard $E$.
\begin{definition}[\<@GuardedBy> for Local Variables] 
\label{def:gardedByValueVariables}
A local variable $x$ of a method $\kappa.m$ in a program is \emph{value-protected} by
{\em $\guardedby{E}$} if and only if
for any derivation $\anglepair{T_0}{\mu_0}\reduces{n_0}\cdots\reduces{n_{i-1}} 
\anglepair{T_{i}}{\mu_{i}} \reduces{n_{i}} \ldots$, letting 
\begin{itemize}
\item $T^n_j=\thread{\record{k^n_j.m^n_j}{C^n_j}{\sigma^n_j}\stsep S^n_j}{\lockset^n_j}$ be the $n$-th thread of the pool $T_j$, for $j > 0$
\item $\mathcal{L}=\bigcup_{j > 0}\{\sigma^{{n_j}}_j(x) \: \mid \: 
  k^{{n_j}}_j.m^{{n_j}}_j{=}\kappa.m\text{ and }\sigma^{{n_j}}_j(x)\!\downarrow \}$ be the set of locations eventually associated to variable $x$
\item 
  $\mathcal{X}_i=\mathit{derefloc}(C^{n_i}_{i})^{\mu_{i}}_{\sigma^{n_i}_{i}}\cap\mathcal{L}$
  be those locations in $\mathcal{L}$ dereferenced at step $\reduces{n_i}$.
\end{itemize}
Then, for every $l\in\mathcal{X}_i$ it follows that 
$\loc{\eval{E}{\mu_{i}}{\sigma^{n_i}_{i}[\mathtt{itself}\mapsto l]}} \in\lockset^{n_i}_{i}$. 
\end{definition}
Note that $\mathcal{L}$ contains all locations eventually bound to $x$, also in the past,
  not just those bound to $x$ in the last configuration $\anglepair{T_i}{\mu_i}$. This is because
  the value of $x$ might change during the execution of the program
  and flow through aliasing into other variables, that later get dereferenced.
\begin{example}
\label{Ex:gb-value-variable}
In Ex.~\ref{ex:simple_semantics} variable \<z> is value-protected by $\guardedby{\itself}$.
The set $\mathcal{X}$ for \<z> of Def.~\ref{def:gardedByValueVariables} is $\{l_1\}$. 
Location $l_1$ is only dereferenced at macro-step $5$, where 
the corresponding object $o_1$ is accessed to obtain the value of its field $f$.
At that program point, location $l_1$ is locked by the current thread.
\end{example}
\begin{definition}[\<@GuardedBy> for Fields]
\label{def:gardedByValueFields}
A field $f$ in a program is \emph{value-protected} by {\em $\guardedby{E}$}
if and only if for any derivation
$
  \anglepair{T_0}{\mu_0} \reduces{n_0} \cdots \reduces{n_{i-1}}
  \anglepair{T_i}{\mu_i} \reduces{n_{i}}\cdots,
$
letting
\begin{itemize}
\item $T^n_j=\thread{\record{k^n_j.m^n_j}{C^n_j}{\sigma^n_j}\stsep S^n_j}{\lockset^n_j}$ be the $n$-th thread of the pool $T_j$, for $j > 0$
\item $\mathcal{L}=\bigcup_{j>0}\{\,\state{\mu_j(l)}(f) \, \mid \, l{\in}\dom{\mu_j} 
\text{ and }\state{\mu_j(l)}(f)\!\downarrow \}$ be the set of locations 
eventually associated to field $f$
\item 
  $\mathcal{X}_i=\mathit{derefloc}(C^{n_i}_{i})^{\mu_{i}}_{\sigma^{n_i}_{i}}\cap\mathcal{L}$ be those locations in $\mathcal{L}$ dereferenced
  at step $\reduces{n_i}$. 
\end{itemize}
Then, for every $l\in\mathcal{X}_i$ it follows that 
$\loc{\eval{E}{\mu_{i}}{\sigma^{n_i}_{i}[\mathtt{itself}\mapsto l]}} \in\lockset^{n_i}_{i}$.
\end{definition}
\begin{example}
In Ex.~\ref{ex:simple_semantics} field \<x> is value-protected by $\guardedby{\itself}$. The set
$\mathcal{X}$ for \<x> of Def.~\ref{def:gardedByValueFields} is $\{l_1\}$
and we conclude as in Ex.~\ref{Ex:gb-value-variable}. 
\end{example}

\begin{remark}
The two semantics for \<@GuardedBy> are incomparable: neither entails the other. For instance, 
in Ex.~\ref{ex:simple_semantics} field \<x> is value protected by \<@GuardedBy(itself)>,
but is not name protected:
\<x> is accessed at macro-step~$1$. Field \<y> is
name protected by \<@GuardedBy(this.x)> but not value protected:
its value is accessed at macro-step~$8$ via \<w>. 
In some cases the two semantics do coincide. Variable
\<z> is \<@GuardedBy(itself)> in both semantics: its name and value
are only accessed at macro-step~$5$, where they are locked.
Variable \<w> is not \<@GuardedBy(itself)> according to any semantics:
its name and value are accessed at macro-step $8$.
\end{remark}

%% file: protection.tex
\section{Protection against Data Races}
\label{subsec:protection}

In this section we provide sufficient conditions that ban
data races when \<@GuardedBy> annotations are satisfied, in either of
the two versions of Sec.~\ref{subsec:byname} and~\ref{sec:gb-value}.
First, we formalize the notion of \emph{data race}.
Informally, a data race occurs when two threads $a$ and $b$ 
dereference the same location $l$, at the same time, to access a field
of the object stored at $l$ and at least one of them, say $a$, modifies the field.
We formalize below this definition for our language.
\begin{definition}[Data race]
\label{def:data_race}
Let $\anglepair{T_0}{\mu_0}\trans{}^{\ast}\anglepair{T}{\mu}$,
where $T_i=\thread{\record{k_i.m_i}{C_i}{\sigma_i}\stsep S_i}{\lockset_i}$ is the $i$-th thread of $T$. 
A \emph{data race} occurs at a location $l$ at $ \anglepair{T}{\mu}$
during the access to a field $f$ if there are $a\not=b$ such that
$\anglepair{T}{\mu} \reduces{a} \anglepair{T'}{\mu'}$,
$\anglepair{T}{\mu} \reduces{b} \anglepair{T''}{\mu''}$,
$l.\mathord{f\!\!\leftarrow}\in\rfrc{C_a}{\mu}{\sigma_a}$ and
($l.\mathord{f\!\!\leftarrow}\in\rfrc{C_b}{\mu}{\sigma_b}$ or $l.\mathord{f\!\!\rightarrow}\in\rfrc{C_b}{\mu}{\sigma_b}$).
\end{definition}

In Sec.~\ref{sec:examples} 
we  said  that accesses to  variables (and fields)
that are \<@GuardedBy($E$)> occur in mutual exclusion
if the guard $E$ is such that it can be evaluated at distinct program points and  its evaluation 
  always yields the same
  value. This means that $E$ cannot contain local variables as they 
  cannot be evaluated at distinct program points. 
Thus, we restrict the variables that can be used in $E$.
In particular, \<itself> can always be used since it refers to the 
location being
dereferenced. For the name-protection semantics for fields, \<this> can also be used,
since it refers to the container of the guarded field,
as long as it
can be uniquely determined; for instance, if there is no aliasing.
Indeed, Sec.~\ref{sec:examples} shows that name protection
without aliasing restrictions
does not ban data races, since it protects the
name but not its value, that can be freely aliased and accessed through other names, without
synchronization.
In a real programming language, \emph{aliasing} arises from
assignments, returned values, and parameter passing.
Our simple language has no returned values and only
the implicit parameter \<this>. 
\begin{definition}[Non-aliased variables and fields]
  \label{def:aliasing}
  Let $P$ be a program and $x$ a variable or field name.
  We say that a name $x$ is \emph{non-aliased} in $P$ if and only if for
  every arbitrary trace $\anglepair{T_0}{\mu_0} \reduces{}^{\ast} \anglepair{T}{\mu}$
  of $P$, where 
  $T_i=\thread{\record{k_i.m_i}{C_i}{\sigma_i}\stsep S_i}{\lockset_i}$ is the $i$-th thread of $T$,
  we have
\begin{itemize} 
\item whenever $\sigma_j(x)=l$, for some $j$ and $l$:
\begin{itemize}
\item there is no $y$, $y\neq x$, such that $\sigma_j(y)=l$
\item there is no $k$, $k\neq j$, such that $\sigma_k(y)=l$, for some $y$
\item there is no $l'$ such that $\state{\mu(l')}(y)=l$, for some $y$
\end{itemize}
\item whenever $\state{\mu(l')}(x)=l$, for some $l'$ and $l$:
\begin{itemize}
\item there are no $y$ and $j$ such that $\sigma_j(y)=l$
\item there are no $y$ and $l''$, $l'\neq l''$, such that $\state{\mu(l'')}(y)=l$.
\end{itemize}
\end{itemize}
\end{definition}
Checking if a name is non-aliased can be mechanized~\cite{Andersen94}
and prevented by syntactic restrictions.
Now, everything is in place to prove that, 
for non-aliased names,
the name-protection semantics of \<@GuardedBy> protects
against data races.
\begin{theorem}[Name-protection semantics 
vs.\ data race protection]
\label{prop:guarantee_by_name}
Let $E$ be an expression in a program, and 
$x$ be a \emph{non-aliased} variable or field that is name protected by {\em $\guardedby{E}$}.
If $x$ is a variable, let $E$ contain no variable distinct from {\em \<itself>\/};
if $x$ is a field, let $E$ contain no variable distinct from {\em \<itself>} and {\em \<this>\/}.
Then, no data race can occur at those locations bound to $x$, at
any execution trace of that program.
\end{theorem}

The absence of aliasing is not necessary for the value-protection semantics. 
\begin{theorem}[Value-protection semantics  vs.\ data race protection]
\label{prop:guarantee}
Let $E$ be an expression in a program, and
$x$ be a variable/field that is value-protected by {\em \<@GuardedBy($E$)>.}
Let $E$ have no variable distinct from {\em \<itself>\/}.
Then no data race can occur at those locations bound to $x$, during 
any execution of the program. 
\end{theorem}
Both results are proved by contradiction, by supposing that a data race 
occurs and showing that
two threads would lock the same location, against
Prop.~\ref{prop:lock-threads}

%% file: julia.tex
\section{Implementation in Julia}\label{sec:julia}
The Julia static analyzer infers \<@GuardedBy> annotations. The user
selects the name-protection or the value-protection semantics.
As discussed in Sec.~\ref{sec:examples}, and then formalized in
Sec.~\ref{sec:guardedBy}, a \<@GuardedBy($E$)> annotation holds for
a variable or field $x$ if, at all program points $P$
where $x$ is accessed
(for name protection) or one of its locations is dereferenced
(for value protection), the value of $E$ is locked by the current thread.
The inference algorithm of Julia builds on two phases:
%
(i) compute $P$; 
(ii) find expressions $E$ locked at all program points in $P$. 

Point (i) is obvious for name protection, since accesses to $x$
are syntactically apparent in the program. For value protection,
the set $P$ is instead undecidable, since there might be
infinitely many objects potentially bound to $x$ at runtime,
that flow through aliasing.
Hence Julia overapproximates $P$ by abstracting objects into their
\emph{creation point} in the program: if two objects have
distinct creation points, they must be distinct. The number of creation
points is finite, hence the approximation is finitely computable.
Julia implements creation points analysis as a concretization of the
class analysis in~\cite{PalsbergS91}, where objects are abstracted in their
creation points instead of just their class tag.

Point (ii) uses the \emph{definite aliasing} analysis of Julia,
described in~\cite{NikolicS12}. At each \<synchronized($G$)> statement,
that analysis provides a set $L$ of expressions that are definitely an alias
of $G$ at that statement
(\emph{i.e.}, their values coincide there, always).
Julia concludes that the expressions
in $L$ are locked by the current thread after the
\<synchronized($G$)> and until the end of its scope. Potential
side-effects might however invalidate that conclusion, possibly due
to concurrent threads. Hence, Julia only allows in $L$
fields that are never modified
after being defined, which can be inferred syntactically for a field.
For name protection, viewpoint adaptation of \<this> is
performed on such expressions (Def.~\ref{def:guardedByFields}).
These sets $L$ are propagated in the program until they reach the points in
$P$. The expressions $E$ in point (ii) are hence those that belong to $L$ at
\emph{all} program points in $P$.

Since \<@GuardedBy($E$)> annotations are expected to be used by client
code, $E$ should be visible to the client. For instance, Julia
discards expressions $E$ that refer to a private field or to a local
variable that is not a parameter, since these would
not be visible nor useful to a client.

The supporting creation points and definite aliasing analyses are
sound, hence Julia soundly
infers \<@GuardedBy($E$)> annotations that satisfy the
formal definitions in Sec.~\ref{sec:guardedBy}. Such inferred
annotations protect against data races if the sufficient conditions
in Sec.~\ref{subsec:protection} hold for them.

More detail and experiments with this implementation can be found
in~\cite{ErnstLMST2015}.
\vspace*{-3mm}

%% file: conclusion.tex
\section{Conclusions, Future and Related Work}\label{sec:conclusion}
\enlargethispage{.75\baselineskip}
We have formalized two possible semantics for Java's \<@GuardedBy> annotations.
Coming back to the ambiguities sketched in Sec.~\ref{sec:introduction},
we have clarified that:
(1) \<this> in the guard expression must be interpreted as the container of the
guarded field and consistently contextualized (Def.~\ref{def:guardedByFields}).
(2) An access is a variable/field use for name protection
(Def.~\ref{def:access}, \ref{def:guardedByVars}, and~\ref{def:guardedByFields}).
A value access is a dereference (field get/set or
method call) for value protection;
copying a value is not an access in this case
(Def.~\ref{def:deref}, \ref{def:gardedByValueVariables},
and~\ref{def:gardedByValueFields}).
(3) The value of the guard expression must be locked when a name or value
is accessed, regardless of how it is accessed for locking
(Def.~\ref{def:guardedByVars}, \ref{def:guardedByFields}, \ref{def:gardedByValueVariables},
and~\ref{def:gardedByValueFields}).
(4) The lock is taken on the value of the guard expression as evaluated
at the access to the
guarded variable or field (Def.~\ref{def:guardedByVars}, \ref{def:guardedByFields},
\ref{def:gardedByValueVariables}, and~\ref{def:gardedByValueFields}
and rule \textrm{[sync]}).
(5) Either the \emph{name} or the
\emph{value} of a variable can be guarded, but this choice leads
to very different semantics.
Namely, in the \emph{name-protection} semantics,
the lock must be held whenever the variable/field's name is accessed
(Def.~\ref{def:access}, \ref{def:guardedByVars}, and~\ref{def:guardedByFields}).
In the \emph{value-protection} semantics, the lock must be held whenever
the variable/field's value is accessed (Def.~\ref{def:deref}, \ref{def:gardedByValueVariables},
and~\ref{def:gardedByValueFields}),
regardless of what expression is used to access the value.
Both semantics yield a guarantee against data races, 
though name protection requires an aliasing restriction
(Th.~\ref{prop:guarantee_by_name} and~\ref{prop:guarantee}).
\looseness=-1

This work could be extended by enlarging the set of guard expressions that
protect against data races. In particular, we have found that programmers often
use \<this> in guard expressions, but in that case we have a proof
of protection only for the name-protection semantics at the moment.
Our simple language already admits local variables and global variables (in object-oriented
  languages, these are the fields of the objects). It could be further extended with
  static fields. We believe that 
  the protection results in Sec.~\ref{subsec:protection} still hold for them.
Another aspect to investigate is the scope
of the protection against data races. In this article, a single location is protected
(Def.~\ref{def:data_race}), not the whole tree of objects reachable from it:
our protection is shallow rather than deep. Deep protection is possibly more interesting
to the programmer, since it relates to a data structure as a whole, but it requires to
reason about boundaries and encapsulation of data structures.

\enlargethispage{1.25\baselineskip}

There are many other formalizations of the syntax and semantics of concurrent Java,
such as~\cite{Abraham-MummBRS02,CenciarelliKRW97}. There is a formalization that
also includes extensions to Java such as RMI~\cite{Nobuko07}. Our goal here is
the semantics of annotations such as \<@GuardedBy>. Hence we kept the
semantics of the language to the
minimum core needed for the formalization of those program annotations.
Another well-known formalization is Featherweight Java~\cite{IgarashiPW01},
a functional language that provides a formal kernel of
sequential Java. It does not include threads, nor assignment.
Thus, it is
not adequate to formalize data races, which need concurrency and assignments.
Middleweight Java~\cite{BiermanP03}
is a richer language, with states, assignments and object identity.
It is purely sequential, with no threads, and its formalization
is otherwise at a level of detail that is unnecessarily complex for the present work.
Welterweight Java~\cite{OstlundW10} is a formalization of a kernel of Java that includes
assignments to mutable data and threads. Our formalization is similar to theirs,
but it is simpler since we do not model aspects that are not relevant
to the definition of data races, such as subtyping.
The need of a formal specification for reasoning about Java's concurrency and for building
verification tools is recognized~\cite{CorbettDH02,LongL03,BogdanasRosu15} but we are
not aware of any formalization
of the semantics of Java's concurrency annotations. Our formalization will support
tools based on model-checking such as Java PathFinder~\cite{JPF} and
Bandera~\cite{HatcliffD01,Bandera}, on type-checking
such as the Checker Framework~\cite{DietlDEMS11}, or on
abstract interpretation such as Julia~\cite{Julia}.
Finally, our companion paper~\cite{ErnstLMST2015} presents the details of the implementation of the Julia analyzer and of a type-checker for \<@GuardedBy> annotations, together with extended experiments that show how these tools scale to large real software and provide useful results for programmers.


%% file: semantics_properties.tex
\section{Properties of the Operational Semantics}\label{sec:properties}
Let us provide a few properties showing the soundness 
of both the locking and unlocking mechanisms of our operational semantics. 

Two different threads never lock the same location:
\begin{proposition}[Locking vs.\ multithreading] 
Given an arbitrary execution trace 
\[
\anglepair{T_0}{\mu_0}\trans{}^*\anglepair{\thread{S_1}{\lockset_1}\parallel
\ldots\parallel \thread{S_n}{\lockset_n}\,}{\mu}
\]
then for any $i,j\in\{1\ldots n\}$, $i\neq j$ entails $\lockset_i\cap\lockset_j = \emptyset$.
\end{proposition}
\begin{proof}
By induction on the length of the trace.
\end{proof}

When a thread starts its execution  it does not hold any lock:
\begin{proposition}[Thread initialization vs.\ locking ]
Let 
\[
\anglepair{T_0}{\mu_0}\trans{}^*\anglepair{\thread{S_1}{\lockset_1}\parallel
\ldots\parallel \thread{S_n}{\lockset_n}\,}{\mu}
\trans{i} \anglepair{\thread{\hat{S}_1}{\hat{\lockset}_1}\parallel
\ldots\parallel \thread{\hat{S}_m}{\hat{\lockset}_m}\,}{\hat{\mu}}
\]
be an arbitrary trace where 
 $S_i = \record{\kappa.m}{\fork{E.p(\,)};C}{\sigma} \stsep S$,
 for some  $\kappa,m,E,p,C,\sigma$ and $S$, 
 then 
\begin{itemize}
\itemsep.13em
\item 
$\hat{\sigma}_i = \record{\kappa'.p}{B}{\sigma'}$, for appropriate $\kappa'$, $B$ and $\sigma'$
\item 
$\hat{\sigma}_{i+1} = \record{\kappa.m}{C}{\sigma} \stsep S$
\item 
$\hat{\lockset}_i = \emptyset$
\item 
$\hat{\lockset}_{i+1} = \lockset_i$. 
\end{itemize}
\end{proposition}
\begin{proof} This is a direct conseuquence of the sematic rules 
[spawn], the only one which can be applied to perform the reduction 
step $\trans{i}$.
\end{proof}

When a thread terminates it does not keep locks on locations:
\begin{proposition}[Thread termination vs.\ locking ]
Let 
\[
\anglepair{T_0}{\mu_0}\trans{}^*\anglepair{\thread{S_1}{\lockset_1}\parallel
\ldots\parallel \thread{S_n}{\lockset_n}\,}{\mu}
\]
be an arbitrary run where 
 $S_i = \epsilon$, , then ${\lockset}_i = \emptyset$.
\end{proposition}
\begin{proof}
By  induction on the lenght of the reduction. 
\end{proof}

A thread may not lock a location by mistake:
\begin{proposition}[Locking]
Let 
\[
\anglepair{T_0}{\mu_0}\trans{}^*\anglepair{\thread{S_1}{\lockset_1}\parallel
\ldots\parallel \thread{S_n}{\lockset_n}\,}{\mu}
\trans{i} \anglepair{\thread{\hat{S}_1}{\hat{\lockset}_1}\parallel
\ldots\parallel \thread{\hat{S}_m}{\hat{\lockset}_m}\,}{\hat{\mu}}
\]
be an arbitrary run. 
Then  $ \bigcup_{j=1}^n{{\lockset}_j} \, \subset \, \bigcup_{j=1}^m{\hat{\lockset}_j}$,  if and only if 
\begin{itemize}
\itemsep.13em
\item 
 $S_i = \record{\kappa.m}{\lock{l};C'}{\sigma} \stsep S$,
 for some $\kappa,m,l,C',\sigma$ and $S$
\item 
 $\lockson{\mu(l)}=0$ and $\lockson{\hat{\mu}(l)}=1$
\item 
 $\hat{\lockset}_i = \lockset_i \uplus \{l \}$
\item
 $m=n$ and $\hat{\lockset}_j = \lockset_j$ for every $j\in\{1\ldots n\}\setminus\{i\}$.
\end{itemize}
\end{proposition}
%
%

Reentrant locks are allowed: only threads that already own the 
lock on an object can synchronize again on that object. 
\begin{proposition}[Reentrant locking]
Given an arbitrary run
\[
\anglepair{T_0}{\mu_0}\trans{}^*\anglepair{\thread{S_1}{\lockset_1}\parallel
\ldots\parallel \thread{S_n}{\lockset_n}\,}{\mu}
\]
where $l \in  \bigcup_{j=1}^n{{\lockset}_j}$, for some $l$, and 
 $S_i = \record{\kappa.m}{\lock{l};C'}{s} \stsep S$,
 for some $i \in \{ 1..n \}$, $\kappa,m,C,E,C',\sigma$ and $S$. 
Then 
\[
\anglepair{\thread{S_1}{\lockset_1}\parallel
\ldots\parallel \thread{S_n}{\lockset_n}\,}{\mu}
\trans{i} \anglepair{\thread{\hat{S}_1}{\hat{\lockset}_1}\parallel
\ldots\parallel \thread{\hat{S}_n}{\hat{\lockset}_n}\,}{\hat{\mu}}
\]
if and only if 
\begin{itemize}
\itemsep.13em
\item 
$l \in \lockset_i $
\item  $\lockson{\hat{\mu}(l)}=\lockson{\mu(l)}+1$
\item 
 $m=n$ and  $\hat{\lockset}_j = \lockset_j$ for every $j\in\{1\ldots n\}$.
\end{itemize}
\end{proposition}
\begin{proof}
By case analysis on the rule applied to perform the reduction.
Here the only rule which can be applied is [reeentrant-lock].
\end{proof}
%
%

Locks on locations are never released by mistake:
\begin{proposition}[Lock releasing]
Let 
\[
\anglepair{T_0}{\mu_0}\trans{}^*\anglepair{\thread{S_1}{\lockset_1}\parallel
\ldots\parallel \thread{S_n}{\lockset_n}\,}{\mu}
\trans{i} \anglepair{\thread{\hat{S}_1}{\hat{\lockset}_1}\parallel
\ldots\parallel \thread{\hat{S}_m}{\hat{\lockset}_m}\,}{\hat{\mu}}
\]
be an arbitrary run. 
Then  $\bigcup_{j=1}^n{{\lockset}_j} \, \supset \, \bigcup_{j=1}^m{\hat{\lockset}_j} $,  if and only if
\begin{itemize}
\itemsep.13em
\item 
 $S_i = \record{\kappa.m}{\unlock{l};C}{\sigma} \stsep S$,
 for some $\kappa,m,l,C,\sigma$ and $S$
\item
  $\lockset_i = \hat{\lockset}_i \uplus \{l \}$ 
\item
 $\lockson{{\mu}(l)} = 1 $ and 
  $\lockson{\hat{\mu}(l)} = 0 $
\item
 $m=n$ and $\hat{\lockset}_j = \lockset_j$ for every $j\in\{1\ldots n\}\setminus\{i\}$.
\end{itemize}
\end{proposition}
\begin{proof}
By case analysis on the rule applied to perform the reduction. 
Here the only possible rule is [release-lock].
\end{proof}
%

Unlocking always happens after some locking: it may release the lock 
or not, depending on the number of previous lockings.
\begin{proposition}[Unlocking]
Let 
\[
\anglepair{T_0}{\mu_0}\trans{}^*\anglepair{\thread{S_1}{\lockset_1}\parallel
\ldots\parallel \thread{S_n}{\lockset_n}\,}{\mu}
\trans{i} \anglepair{\thread{\hat{S}_1}{\hat{\lockset}_1}\parallel
\ldots\parallel \thread{\hat{S}_m}{\hat{\lockset}_m}\,}{\hat{\mu}}
\]
be an arbitrary run where 
 $S_i = \record{\kappa.m}{\unlock{l};C}{\sigma} \stsep S$,
 for some $\kappa,m,C,\sigma$ and $S$, then 
\begin{itemize}
\itemsep.13em
\item 
$l \in {\lockset}_i$
\item if $\lockson{{\mu}(l)} > 1$ then $\hat{\lockset}_i = \lockset_i$ 
else  $\lockset_i = \hat{\lockset}_i \uplus \{l \}$ 
\item  
$\lockson{\hat{\mu}(l)} = \lockson{\mu(l)} -1 $
\item
 $m=n$ and $\hat{\lockset}_j = \lockset_j$ for every $j\in\{1\ldots n\}\setminus\{i\}$.
\end{itemize}
\end{proposition}
\begin{proof}
By case analysis on the rules applied to perform the reduction. 
Here there are two possible rules: [decrease-lock] and [release-lock].
\end{proof}

%% file: main.bbl
\begin{thebibliography}{10}

\bibitem{Abraham-MummBRS02}
E.~{\'{A}}brah{\'{a}}m{-}Mumm, F.~S. de~Boer, W.~P. de~Roever, and M.~Steffen.
\newblock Verification for {J}ava's {R}eentrant {M}ultithreading {C}oncept.
\newblock In {\em FOSSACS}, pages 5--20, Grenoble, France, 2002.

\bibitem{Nobuko07}
A.~Ahern and N.~Yoshida.
\newblock Formalising {J}ava {RMI} with {E}xplicit {C}ode {M}obility.
\newblock {\em Theoretical Computer Science}, 389:341--410, 2007.

\bibitem{Andersen94}
L.~O. Andersen.
\newblock {\em Program {A}nalysis and {S}pecialization for the {C}
  {P}rogramming {L}anguage}.
\newblock PhD thesis, University of Copenhagen, DIKU, 1994.

\bibitem{Bandera}
Bandera.
\newblock About {B}andera.
\newblock \url{http://bandera.projects.cis.ksu.edu}.

\bibitem{BiermanP03}
G.~M. Bierman and M.~J. Parkinson.
\newblock Effects and {E}ffect {I}nference for a {C}ore {J}ava {C}alculus.
\newblock In {\em WOOD}, pages 82--107, 2003.

\bibitem{Blanchet03}
B.~Blanchet.
\newblock Escape {A}nalysis for {J}ava: 
  and {P}ractice.
\newblock {\em ACM TOPLAS}, 25(6):713--775, 2003.

\bibitem{BogdanasRosu15}
D.~Bogdanas and G.~Rosu.
\newblock {K}-java: {A} {C}omplete {S}emantics of {J}ava.
\newblock In {\em {ACM} {SIGPLAN-SIGACT} {POPL}}, pages 445--456, Mumbai,
  India, 2015.

\bibitem{CenciarelliKRW97}
P.~Cenciarelli, A.~Knapp, B.~Reus, and M.~Wirsing.
\newblock From {S}equential to {M}ulti-{T}hreaded {J}ava: {A}n {E}vent-{B}ased
  {O}perational {S}emantics.
\newblock In {\em AMAST}, pages 75--90, Sydney, AU, 1997.

\bibitem{CorbettDH02}
J.~C. Corbett, M.~B. Dwyer, J.~Hatcliff, and Robby.
\newblock Expressing {C}heckable {P}roperties of {D}ynamic {S}ystems: the
  {B}andera {S}pecification {L}anguage.
\newblock {\em STTT}, 4(1):34--56, 2002.

\bibitem{DietlDEMS11}
W.~Dietl, S.~Dietzel, M.~D. Ernst, K.~Muslu, and T.~W. Schiller.
\newblock Building and {U}sing {P}luggable {T}ype-{C}heckers.
\newblock In {\em ICSE}, pages 681--690, Waikiki, Honolulu, HI, USA, 2011.

\bibitem{DietlDM2007}
W.~Dietl, S.~Drossopoulou, and P.~M{\"u}ller.
\newblock {G}eneric {U}niverse {T}ypes.
\newblock In {\em ECOOP}, pages 28--53, Berlin, Germany, August 2007.

\bibitem{GoetzPBB06}
B.~Goetz, T.~Peierls, J.~Bloch, and J.~Bowbeer.
\newblock {\em Java {C}oncurrency in {P}ractice}.
\newblock Addison Wesley, May 2006.

\bibitem{Guava}
Google.
\newblock Guava: {G}oogle {C}ore {L}ibraries for {J}ava 1.6+.
\newblock \url{https://code.google.com/p/guava-libraries}.

\bibitem{HatcliffD01}
J.~Hatcliff and M.~B. Dwyer.
\newblock Using the {B}andera {T}ool {S}et to {M}odel-{C}heck {P}roperties of
  {C}oncurrent {J}ava {S}oftware.
\newblock In {\em CONCUR}, pages 39--58, Aalborg, Denmark, August 2001.

\bibitem{IgarashiPW01}
A.~Igarashi, B.~C. Pierce, and P.~Wadler.
\newblock Featherweight {J}ava: {A} {M}inimal {C}ore {C}alculus for {J}ava and
  {GJ}.
\newblock {\em ACM TOPLAS}, 23(3):396--450, 2001.

\bibitem{JavadocGuardedBy}
Javadoc for \<@{G}uarded{B}y>.
\newblock
  \url{https://jsr-305.googlecode.com/svn/trunk/javadoc/javax/annotation/concurrent/GuardedBy.html}.

\bibitem{LongL03}
B.~Long and B.~W. Long.
\newblock Formal {S}pecification of {J}ava {C}oncurrency to {A}ssist {S}oftware
  {V}erification.
\newblock In {\em IPDPS}, page 136, Nice, France, April 2003. {IEEE} Computer
  Society.

\bibitem{JPF}
{NASA}.
\newblock Java {P}ath{F}inder.
\newblock \url{http://babelfish.arc.nasa.gov/trac/jpf}.

\bibitem{NikolicS12}
D.~Nikolic and F.~Spoto.
\newblock Definite {E}xpression {A}liasing {A}nalysis for {J}ava {B}ytecode.
\newblock In {\em ICTAC}, pages 74--89, Bangalore, India, September 2012.

\bibitem{OstlundW10}
J.~{\"{O}}stlund and T.~Wrigstad.
\newblock Welterweight {J}ava.
\newblock In {\em TOOLS}, pages 97--116, M{\'{a}}laga, Spain, 2010.

\bibitem{PalsbergS91}
J.~Palsberg and M.~I. Schwartzbach.
\newblock Object-{O}riented {T}ype {I}nference.
\newblock In {\em OOPSLA}, volume 26(11) of {\em ACM SIGPLAN Notices}, pages
  146--161. ACM, November 1991.

\bibitem{Pech10}
V.~Pech.
\newblock Concurrency is {H}ot, {T}ry the {JCIP} {A}nnotations.
\newblock \url{http://jetbrains.dzone.com/tips/concurrency-hot-try-jcip},
  February 2010.

\bibitem{Plo81}
G.D. Plotkin.
\newblock A {S}tructural {A}pproach to {O}perational {S}emantics.
\newblock Technical Report {DAIMI-FN}-19, Computer Science Department, Aarhus
  University, 1981.

\bibitem{Julia}
Julia Srl.
\newblock The {J}ulia {S}tatic {A}nalyzer.
\newblock \url{http://www.juliasoft.com/julia}.

\end{thebibliography}
